\documentclass[prodmode,acmec]{ec-acmsmall}
\acmVolume{X}
\acmNumber{X}
\acmArticle{X}
\acmYear{2015}
\acmMonth{2}
\usepackage[numbers]{natbib}
\usepackage{graphicx}
\usepackage{amsmath}
\usepackage{amssymb}
\usepackage[all]{xy}
\usepackage{xspace}
\usepackage[tight]{subfigure}

\newtheorem{question}[theorem]{Question}
\newtheorem{observation}[theorem]{Observation}
\newtheorem{condition}[theorem]{Condition}

\newcommand{\xhdr}[1]{\paragraph*{\bf #1}}
\newcommand{\yhdr}[1]{\subsection {\bf #1}}

\newcommand{\Exp}[1]{
E\left[{#1}\right]
}

\newcommand{\omt}[1]{}
\newcommand{\Xomit}[1]{}
\newcommand{\supproof}[1]{}
\newcommand{\supproofeq}[1]{}

\newcommand{\twoc}[1]{}

\begin{document}
\title{Team Performance with Test Scores}
  
  \author{
  Jon Kleinberg   \affil{Cornell University, Ithaca NY}
  Maithra Raghu \affil{Cornell University, Ithaca NY}
  }
  
  \begin{bottomstuff}
  Authors' emails:
  {\tt kleinber@cs.cornell.edu},
  {\tt mraghu@cs.cornell.edu}.
  \end{bottomstuff}

\begin{abstract}
Team performance is a ubiquitous area of inquiry in the social
sciences, and it motivates
the problem of {\em team selection} --- choosing the members of a team
for maximum performance.
Influential work of Hong and Page has argued that testing individuals
in isolation and then assembling the highest-scoring ones into a team
is not an effective method for team selection.
For a broad class of performance measures, based on the expected maximum of random variables representing individual candidates,
we show that tests directly measuring individual performance are
indeed ineffective,
but that a more subtle family of tests used in isolation can
provide a constant-factor approximation for team performance.
These new tests measure the ``potential'' of individuals, in a precise
sense, rather than performance; to our knowledge they represent the
first time that individual tests have been shown to produce near-optimal
teams for a non-trivial team performance measure.
We also show families of subdmodular and supermodular
team performance functions for which {\em no} test applied to
individuals can produce near-optimal teams, and
discuss implications for submodular maximization
via hill-climbing.

\end{abstract}

\maketitle

\section{Introduction}

The performance of teams in solving problems has been a subject of
considerable interest in multiple areas of the mathematical social sciences
\cite{gully-teams-meta-analysis,kozlowski-team-performance,wuchty-teams-knowledge}.
The ways in which groups of people come together and accomplish tasks
is an important issue in theories of organizations, innovation, and
other collective phenomena, and the recent growth of interest
in crowdwork has brought these issues into focus for on-line platforms as well.

In formal models of team performance, a central issue is the problem
of {\em team selection}.  Suppose there is a task to be accomplished and
we can assemble a team to collectively work on this task, drawing
team members from a large set $U$ of $n$ candidates.  
(We can think of $U$ as the job applicants for this task.)  
A team can be any subset $T \subseteq U$, and its
performance in collectively working on the task is given by a set
function $g(T)$.  The central optimization problem is therefore a kind of
set function maximization: given a target size $k < n$ for the team,
we would like to find a set $T$ of cardinality $k$ for which $g(T)$
is as large as possible.

The generality of this framework has meant that it can be used
to reason about a wide range of settings in which we hire workers,
solicit advice from a committee, 
run a crowdsourced contest,
admit college applicants, and many other activities
--- all cases where we have an objective function (the outcome
of the work performed, the quality of the insights obtained,
or reputation of the group that is assembled) that is a function of
the set of people we bring together.

\xhdr{Models of Team Performance}
Different models of team performance can be interpreted as positing
different forms for the structure of the set function $g(\cdot)$.
Some of the most prominent have been the following.
\begin{itemize}
\item {\em Cumulative effects.}  Arguably the simplest team performance
function is a linear one: each individual can produce work at a certain
volume, and the team's performance is simply the sum of these
individual outputs.  
Formally, we assume that each individual $i \in U$ has
a weight $w_i$, and then $g(T) = \sum_{i \in T} w_i$.
\item {\em Contests.}  Much work has focused on models of
team performance in which the ``team'' is highly decoupled:
members attempt the task independently, and the quality of the outcome
is the maximum quality produced by any member.
Such formalisms arise in the study of contest-like processes, where
many competitors independently contribute proposed solutions, and
a coordinator selects the best one
(or perhaps the $h$ best for some $h < k$) 
\cite{jeppesen-team-contest,lakhani-sci-contests}.
Note however that this objective function is applicable more generally
to any setting with a ``contest structure,'' even potentially
inside a single organization, where proposed solutions are generated 
independently and the outcome is judged by the quality of the best one
(or best few).
It can also apply to a group whose reputation is judged on the
maximum future achievement of any of its members; for example, one could
imagine an admissions committee trying to select a group of
$k$ top applicants, with the goal of optimizing the maximum future success
of any of them.
\item {\em Complementarity.}  Related to contests are models in which
each team member has a set of ``perspectives,'' and the quality of the
team's performance grows with the number of distinct perspectives
that they are collectively able to provide 
\cite{hong-page-diversity-pnas,marcolino-tambe-team-formation}.
\item {\em Synergy.}  In a different direction, research has also
considered models of team performance in which interaction
is important, using objective functions with terms that generate 
value from pairwise interaction between team members
\cite{ballester-key-player}.
\end{itemize}

These settings are not just different in their motivation; they
rely on functions $g(\cdot)$ with genuinely 
different combinatorial properties.
In particular, in the language of set functions,
the first class of instances is based on {\em modular} (i.e. linear)
functions, the second and third classes are based on
{\em submodular} functions, and the fourth is based on
{\em supermodular} functions.

The second and third classes of functions --- contests and complementarity ---
play a central role in Scott Page's highly influential line of work on
the power of diversity in team performance 
\cite{page-difference-book}.
The argument, in essence, is that a group with diversity that is
reflected in independent solutions or complementary perspectives can 
often outperform a group of high-achieving but like-minded members.

\xhdr{Evaluating Team Members via Tests}
A key issue that Page's work brings to the fore is the question
of {\em tests} and their effectiveness in identifying good team members
\cite{page-difference-book}.
In most settings one can't ``preview'' the behavior of a set of team
members together, and so a fundamental approach to team formation
is to give each 
candidate $i \in U$ a {\em test}, resulting in a test score $f(i)$
\cite{miller-team-tests}.
It is natural to then select the $k$ candidates with the highest
test scores, resulting in a team $T$.
We could think of the test score $f(i)$ corresponding to the SAT or GRE
score in the case of college or graduate school admissions, or corresponding
to the quality of answers to a set of technical interview questions
in a job interview.
We note that this issue of tests as a method of selection is
a contribution of Page's work that is related to the issue of
diversity, but also has interesting implications independently
of diversity, and it is the properties of tests that serves
as our focus in the present paper.

Should we expect that the $k$ individuals who score highest on the
test will indeed make the best team?
In a simple enough setting, the answer is yes --- for modular functions
$g(T) = \sum_{i \in T} w_i$, it is enough to evaluate each 
candidate $i$ in isolation, applying the test $f(i) = g(\{i\}) = w_i$.
Let us refer to $f(i) = g(\{i\})$ in general as the {\em canonical test} ---
we simply see how $i$ would perform as a one-element set.
For modular functions, 
clearly the $k$ candidates with the highest scores under the canonical
test form the best team.

On the other hand, Hong and Page construct an example, based on
complementarity, in which the $k$ candidates who score highest on
the canonical test perform significantly worse as a team than
a set of $k$ randomly selected candidates
\cite{hong-page-diversity-pnas}
Their mathematical analysis has a natural interpretation with 
implications for hiring and admissions processes: 
the $k$ candidates who score highest on the test are too
similar to each other, and so with an objective function based
on complementarity, they collectively represent many fewer perspectives
than a random set of $k$ candidates.

Beyond these compelling examples, however, there is very little 
broader theoretical understanding of the power of tests in selecting teams.
Thinking of tests as arbitrary functions of the candidates is
not a perspective that has been present in this earlier work;
a particularly unexplored issue is the fact that the failure of the
canonical test doesn't necessarily rule out the possibility that other
tests might be effective in assembling teams.
Does it ever help, in a formal sense, 
to evaluate a candidate using a measure $f(i)$
that is different from his or her actual individual performance at the task?
In real settings, we see many cases where employers, search committees,
or admissions committees evaluate applicants on their ``potential''
rather than on their demonstrated performance ---
is this simply a practice that has evolved for reasons of its own, or
does it have a reflection in a formal model of team selection?
Without a general formulation of tests as a means for evaluating
team members, it is difficult to offer insights into these basic questions.

\xhdr{The Present Work: Effective Tests for Team Selection}
In this paper we analyze the power of general tests in forming
teams across a range of models.
Our main result is the finding that for team performance measures 
that have a contest structure, near-optimal teams can be selected
by giving each candidate a test in isolation, and then ranking
by test scores, but {\em only} using
tests that are quite different from the canonical test.
To our knowledge, this is the first result to establish that
non-standard tests can yield good team performance in settings
where the canonical test provably fails.

In more detail, in a contest structure each candidate $i \in U$
has an associated discrete random variable $X_i$, with all random variables
mutually independent, and the performance of a team $T \subseteq U$
is the expected value of the random variable $\max_{i \in T} X_i$.
More generally, we may care about the top $h$ values, 
for a parameter $h < k$, in which case the performance of $T$
is the expected value of the sum of the $h$ largest random variables
in $T$: $$g(T) = \Exp{\max_{S \subseteq T, |S| = h} \sum_{i \in S} X_i}.$$

The test that works well for these contest functions has a natural
and appealing interpretation.
Focusing on the general case with parameter $h < k$, we
define the test score $f(i)$ to be 
$$\Exp{\max(X_i^{(1)}, X_i^{(2)}, \ldots, X_i^{(k/h)})},$$
where $X_i^{(1)}, X_i^{(2)}, \ldots, X_i^{(k/h)}$ represent 
$k/h$ independent random variables all with the same distribution as $X_i$.

The fact that this test works for assembling near-optimal teams
in our contest setting has a striking interpretation ---
it provides a formalization of the idea that we should indeed sometimes
evaluate candidates on their potential, rather than their 
demonstrated performance.
Indeed, $\max(X_i^{(1)}, X_i^{(2)}, \ldots, X_i^{(k/h)})$ is precisely
a measure of potential, since instead of just evaluating
$i$'s expected performance $\Exp{X_i}$, we're instead
asking, ``If $i$ were allowed to attempt the task $k/h$ times
independently, what would the best-case outcome look like?''
Like the argument of Hong and Page about diversity, this argument
about potential has qualitative implications for evaluating candidates
in certain settings --- that we should think about upside
potential using a thought experiment in which candidates are
allowed multiple independent tries at a task.

Following this result, we then prove a number of other theorems
that help round out the picture of general tests and their power.
We first show a closely related test that also provides a
method for constructing near-optimal teams, in which 
$f(i)$ is defined to be the conditional expectation of $X_i$,
conditioned on its taking a value in the top $(1/k)$ fraction of
its distribution.
We also show that there exists an absolute constant $c > 1$
such that {\em no} test can construct teams 
under our objective
function with performance guaranteed to come within a factor $c$ 
of optimal.

Next, we show that there are natural objective functions for which no test
can yield near-optimal results for team selection --- these
include certain submodular functions capturing complementarity
and certain supermodular functions representing synergy.
Note that this is a much stronger statement than simply asserting
the failure of the canonical test, since it says that {\em no test}
can produce near-optimal teams.
Finally, we identify some further respects in which 
team performance functions $g(\cdot)$ based on contest structures
have tractable properties, in particular showing that for the
special case in which the random variables corresponding to all the candidates
are weighted Bernoulli variables, greedy hill-climbing on the
value of $g(\cdot)$ in fact
produces an exactly optimal set of size $k$.

\xhdr{The Power of Tests in Competitive Settings} 
Our discussion of test scores can be viewed as pursuing a family of questions of the following general form: ``When evaluating the effectiveness of an individual, to what extent can we perform this evaluation in isolation, and to what extent do we need the context in which they are operating?'' 

This type of question can be asked in settings other than team formation, and in the final section we show how it leads to interesting results if we ask it in a setting with competition between individuals. Specifically, suppose we have a collection of {\em competitors}, and these competitors will be matched up in pairwise competitions. Each competitor $i$ is represented by a random variable $X_i$, representing the distribution of performance quality that $i$ exhibits in competition. When $i$ and $j$ are paired in a competition, we imagine that they draw values independently from $X_i$ and $X_j$ respectively, and the competitor who draws the larger value wins. (We'll say that they tie if the values drawn are equal.) Thus the probability that $X_i$ wins or ties is $\mathbb{P}(X_i \geq X_j)$. 

We'd like to assign each competitor with random variable $X$ a {\em score} $f(X)$, based only on $X$ and not any of the other random variables, so that when two competitors are paired up, the one with the higher score has a reasonably large probability of winning (or tieing). In other words, we'd like to find a function $f$ defined on arbitrary random variables, and an absolute constant $c > 0$, such that if $f(X) \geq f(Y)$, then $\mathbb{P}(X \geq Y) \geq c$. 

Is this possible, and if so, how large can we make $c$? We give a tight answer to this question: the largest possible $c$ is $c = 1/4$. To do this, we first establish $c = 1/4$ can be achieved by the function $f$ that maps each $X$ to its median $f(X)$. We then establish that $c$ cannot be any larger using an argument based on the notion of {\em non-transitive dice}. 

We feel that the emergence of rich questions in this very different domain suggests that there may be other unexpected settings in which an understanding of test scores might lead to interesting insights.

\section{Team Selection by Test Score}
In this section, we formalize our goal of picking individual via a test score to maximize a notion of team performance. We precisely define our measure of team performance, and also define a test that can be applied to individuals for team selections. This test is particularly remarkable, because no matter the size of the team we pick using this test, we can give a constant (independent of team size) order performance guarantee on our test selected team compared to the optimal team. The latter parts of this section build the necessary mathematical tools and definitions needed, and then prove this result. 

In doing so, we build on basic properties of the maximum over sets of random variables, and expect that these results will be useful more broadly.

\subsection{Problem Setting and Key Definitions}
Suppose we are trying to assemble a team of fixed size $k$. We have $N$ possible candidates for this team, each associated with a non-negative discrete random variable $X_i$. Each $X_i$ represents the latent ability of the candidate. For example, if $X_i$ took values $(1, 0.4, 0)$ with probabilities $(0.75, 0.2, 0.05)$, candidate $i$, when put to test, will most likely (with probability $0.75$) perform with skill $1$, and with lower chance (probability $0.2$) perform with skill $0.4$. There is also a small chance (probability $0.05$) that they might perform very poorly, with skill $0$. Setting up notation, we assume each $X_i$ has a distribution $(p_1,...,p_n)$ over nonnegative values $(x_1,...,x_n)$, with $x_1 > x_2,...> x_n \geq 0$. 

To select our team, we can test any of our candidates \textit{individually} but not as a group. Testing a candidate individually corresponds to applying a scoring function $f(X_i)$ to the random variable $X_i$ representing the candidate. We can then rank candidates according to their scores, and pick the top $k$ to form our team. The performance of our team is measured by a \textit{team} scoring function $g$. 

Our work first looks at devising a test function $f$ when the team scoring function $g$ is the \textit{expected maximum}. Having picked our team to comprise of $X_1,...,X_k$, the team performance is given by
\[ g(X_1,...,X_k) = \mathbb{E}(\max\{X_1,...,X_k \}) \]
If the team scoring function is the expected maximum, an immediate first candidate for $f$ might be the expectation, $f(X_i) = \mathbb{E}(X_i)$, which we refer to as the \textit{canonical test}. However, as discussed in Section 3, this first choice is highly suboptimal: we can show that picking a team according to this test results in a multiplicative factor $k$ performance difference between the chosen team and the optimal team. Instead, we define the following, more subtle test. Let $X^{(i)}$ be iid copies of the random variable $X$. Then:
\[ f(X) = \mathbb{E}\left(\max(X^{(1)},...,X^{(k)}) \right) \]
We can interpret $f$ as a better test of the \textit{potential} of $X$, where instead of taking the expectation, we take the best effort when $X$ is given multiple ($k$) attempts. Remarkably, picking a team according to this test results in a \textit{constant factor} (independent of $k$) guarantee on the chosen team's performance compared to the optimal team.

In the following subsections, we build towards and culminate with a proof of this result. In fact, we work with a more general individual test function $f$, and team scoring function $g$:

\begin{definition} \emph{(Team Performance Scoring Function)} \label{def:metric} For (nonnegative) random variables $X_1,...,X_k$, and for $i \leq k$, let $X_{(X_1,...,X_k)}^{(i)}$ denote the $i^{\mathrm{th}}$ largest random variable out of $X_1,...,X_k$. Then for  $1 \leq h \leq k$, let:
$$ g_h(X_1,...,X_k) = \mathbb{E} \left( X_{(X_1,...,X_k)}^{(1)} + X_{(X_1,...,X_k)}^{(2)} + ... + X_{(X_1,...,X_k)}^{(h)} \right) $$
\end{definition}

\begin{definition} \emph{(Individual Testing Function)}
For a nonnegative discrete random variable $X$, and $h \leq k$, let
$$f_h(X) = \mathbb{E}\left(\max(X^{(1)},...,X^{(\frac{k}{h})}) \right) $$
where $X^{(i)}$ denotes an iid copy of $X$.
\end{definition}

These definitions provide a natural interpolation between potential and expected performance. For $h=1$, the team performance function $g$ again becomes the expected maximum, and similarly the individual scoring function $f$ is the corresponding `potential' test function defined earlier. Recall that in this setting, the canonical test (the test of expected performance), is a very poor test for assembling a team. However for $h=k$, the team performance function becomes $\mathbb{E}(\sum X_i)$, and the individual testing function collapses to the canonical test $\mathbb{E}(X_i)$. But as $\mathbb{E}(\sum X_i) = \sum \mathbb{E}(X_i)$, the canonical test is in this case the \textit{perfect} test.

\subsection{Preliminary Mathematical Results: The top $h/2k$ quantile}
In the previous section, we defined our general team performance scoring function (for $h = 1$, the expected maximum and more generally the expectation of the sum of the top $h$ performances of our team of size $k$), and our corresponding individual test function (for $h=1$ the expected maximum of $k$ copies of $X$ and more generally the expectation of the sum of the top $h$ performances of $k$ copies of $X$). 

In this section, we derive important definitions and lemmas to allow us to prove the central result relating team performance when selecting with our test function: the constant factor performance guarantee with respect to the optimal team. Central to all of these is the notion of the \textit{top quantile} of a random variable's distribution. Intuitively speaking, for some proportion $t$, we can define the top $t$ quantile of a discrete random variable to be the largest values taken by the random variable that are responsible for proportion $t$ of its probability mass. Returning to our example of $X_i$ with values $(1, 0.4, 0)$ and probabilities $(0.75, 0.2, 0.05)$, the top $0.6$ quantile of $X_i$ would be $\{1\}$, as $X_i$ takes value $1$ with probability $> 0.6$. The top $0.8$ quantile of $X_i$ would be $\{1, 0.4\}$, as the probability mass of $1$ alone is less than $0.8$, but the probability mass of both values combined is $> 0.8$.

To formalize this, we turn to the notion of a random variable's sample space, treating our random variable $X$ as a function on \textit{events} $\omega \in [0, 1]$. We formalize this in the definition below.\footnote{We note that 
some of our basic definitions can be expressed in the
language of {\em order statistics}, in which we take a set of 
given random variables $X_1, \ldots, X_n$, and a parameter $k$, and 
we construct a new random variable equal to the
$k^{\rm th}$ largest value among $X_1, \ldots, X_n$ 
\cite{david-order-statistics}.
However, for our purposes, the general results about order statistics
do not seem to provide more direct ways of handling any of the constructs
in our analysis, and so we instead use the presentation 
developed in this section.}

\begin{definition} \label{obs: sample_space}
For a nonnegative discrete random variable $X$, we define, for $\omega \in [0,1]$,
$$X(\omega) =
\left\{
	\begin{array}{ll}
		x_1 & \mbox{if } \omega > 1 - p_1 \\
		x_l & \mbox{if }  1 - \sum_{i=1}^l p_i < \omega \leq \ 1 - \sum_{i=l-1}^n p_i \\
		0  & \mbox{if } \omega \leq 1 - \sum_{i=1}^n p_i \\	
	\end{array}
\right.$$
\end{definition}

With this definition, we can also make precise what we mean by the {\em top values} of $X$:

\begin{definition} \label{def: top-values}
For nonnegative discrete $X$ with sample space $[0,1]$, the event $A$ that $X$ takes values in its top $h/2k$ quantile is 
$$ A = \{ \omega : \omega > 1 - \frac{h}{2k} \}$$ 
The {\em top values} of $X$ are then
$$ \{ x_i : 1 \leq i \leq n, \exists \omega \in A, X(\omega) = x_i \}$$
Similarly, we can define the {\em tail values} to be
$$ \{ x_i : 1 \leq i \leq n, \exists \omega \in A^c, X(\omega) = x_i \}$$
\end{definition}
Returning to our example, if $h/2k = 0.6$, then the top values of $X$ would be $\{1\}$, and the tail values would be $\{1, 0.4, 0 \}$. If $h/2k = 0.4$, then the top values would be $\{1, 0.4\}$, and the tail values $\{0.4, 0\}$. Note that there are values that appear in both top and tail in both the top and tail values, and indeed more generally, that the top values and tail values are usually not disjoint -- for the boundary value $x_t$, we may have to split $\{ \omega : X(\omega) = x_t \}$ into $A$ and $A^c$.

%USEFUL SHORTCUTS -----------------------------
\def\stath{{\left(1 - (1 - q_1)^{k/h} \right)x_1 + \left((1 - q_1)^{k/h} - (1 - q_2)^{k/h} \right)x_2 + ... + \left((1 - q_{n-1})^{k/h} - (1 - q_n)^{k/h} \right)x_n}}

\def\stat1{{\left(1 - (1 - q_1)^{k} \right)x_1 + \left((1 - q_1)^{k} - (1 - q_2)^{k} \right)x_2 + ... + \left((1 - q_{n-1})^{k} - (1 - q_n)^{k} \right)x_n}}

\def\truncstat{{\left(1 - (1 - q_1)^{k/h} \right)x_1 +  ... + \left((1 - q_{i-1})^{k/h} - (1 - q_i)^{k/h} \right)x_i}}

\def\truncstat1{{\left(1 - (1 - q_1)^{k} \right)x_1 + ... + \left((1 - q_{i-1})^{k} - (1 - q_i)^{k} \right)x_i}}

\def\mre{{1 - \frac{1}{\sqrt{e}}}}
\def\m4re{{1 - \frac{1}{\sqrt[4]{e}}}}

\def\emax{{ \mathbb{E} \left( \max(X_1,...,X_k) \right)}}
\def\emaxh{{ \mathbb{E} \left( \max(X_1,...,X_{k/h} \right) }}

%END OF USEFUL SHORTCUTS -----------------------

Before proceeding with the lemmas, we make a short comment on notation: from now on, all random variables $X$ are assumed to be discrete and nonnegative, with probabilities $(p_1,...,p_n)$ over values (in decreasing order) $(x_1,...,x_n)$. We define $q_i$ to be the \textit{cumulative sum} of the top $i$ probabilities, i.e. 
\[ q_i = \sum_{l=1}^i p_l \]
We will also often use $(x_1,...,x_t)$ to denote the top values of $X$, with the probability mass associated with $x_t$ split so that $q_t = \frac{h}{2k}$ exactly.

Our first two lemmas rely on the explicit form of our testing function $f_h$. In particular, with the definition of $q_i$, we have:
\[ f_h(X) = ((1 - (1 - q_1)^{k/h})x_1 + ((1 - q_1)^{k/h} - (1 - q_2)^{k/h})x_2 + ... + ((1 - q_{n-1})^{k/h} - (1 - q_n)^{k/h})x_n) \]
In the first two lemmas, we (1) bound the proportion that the top $h/2k$ quantile contributes to $f_h(X)$, (2) upper bound the contribution of the tail values of $X$ to $f_h(X)$. Splitting according to the top $h/2k$ is important as for the main result, we bound $g_h$ by $f_h$ by evaluating the top and tail contributions separately.

\begin{lemma} \label{lemma: topvalbound}
Let $X$ be a random variable, with underlying sample space $[0,1]$. Define $X'$ as
$$ X'(\omega) =
\left\{
	\begin{array}{ll}
		X(\omega)  & \mbox{if } \omega > 1 - \frac{h}{2k} \\
		0 & \mathrm{o/w}
	\end{array}
\right. $$
Then 
$$f_h(X') \geq f_h(X) \left( \mre \right) $$
\end{lemma}

\begin{proof}
First note that if $B$ is the event that some $X^{(i)}$ in the $k/h$ copies of $X$ in $f_h(X)$ takes one of its top values, $(x_1,...,x_t)$, then certainly 
$$f_h(X|B) = \mathbb{E}\left(\max(X^{(1)},...,X^{(k/h)})| B \right) \geq f_h(X) $$
(as we are conditioning on an event concentrated on the highest possible values). But the left hand side can be written out in full as
$$\frac{1}{(1 - (1 - q_t)^{k/h})} \left( \left(1 - (1 - q_1)^{k/h} \right)x_1 + ... + \left((1 - q_{t-1})^{k/h} - (1 - q_t)^{k/h} \right)x_t \right) \geq f_h(X)$$
But this is just
$$ \frac{1}{(1 - (1 - q_t)^{k/h})} \cdot f_h(X')  \geq f_h(X)$$
Noting that $1 - (1 - q_t)^{k/h} \geq ( \mre )$ gives the result.
\end{proof}

We have therefore shown that a transformation mapping $X$ to $X'$, non zero only on the top $h/2k$ quantile of $X$, does not result in too large a loss in the value of $f_h(X)$. 

\begin{lemma} \label{lem: tail-bound}
Let $X$ have $(x_1,...,x_t)$ as its top values, with $q_t = \frac{h}{2k}$. Then
$$ x_l < \frac{f_h(X)}{\mre}$$
for any $ l \geq t$
\end{lemma}

\begin{proof} 
Note that
$$\left( 1 - (1 - q_t)^{k/h} \right)x_t \leq f_h(X) $$
The Lemma then follows by noting that $\left( 1 - (1 - q_t)^{k/h} \right) > \mre $, and that $x_l \leq x_t$ for $l \geq t$.
\end{proof}

Next we prove a simple lemma on certain functions increasing in value, and then invoke this lemma to show that for random variables with total probability mass corresponding to positive values less than $h/2k$, we can bound our test function $f_h$ with respect to the canonical test of expected value, and with respect to a conditional expectation. Again, these lemmas will bound specific parts of bounds relating $f_h$ and $g_h$. 

\begin{lemma} \label{lemma: inc-fns}
For $a \geq 1$, the functions
$$(1 - x)^a - (1 - ax)$$
and
$$\left(1 - \frac{a}{2}x \right) - (1 - x)^a$$
are increasing for $x \in \left[0, \displaystyle{\frac{1}{2a}} \right]$
\end{lemma}

\begin{proof}
Differentiating, and removing the positive factor of $a$, we have
$$1 - (1 - x)^{a - 1}$$
which is $\geq 0$ for $x \in [0,1] $ and 
$$(1 - x)^{a-1} - \frac{1}{2}$$
which achieves its minimum value at $x = \frac{1}{2a}$ but remains nonnegative for $a \geq 1$.
\end{proof}

\begin{lemma} \label{lem:rare_potbound}
For a random variable $X$, with total probability mass for positive values $\leq h/2k$ (i.e. $q_n \leq h/2k$), we have
$$ \frac{hf_h(X)}{k} \leq \mathbb{E}(X) \leq \frac{2hf_h(X)}{k} $$
\end{lemma}

\begin{proof}
$f_h(X)$ can be written explicitly as
$$ \stath $$
Noting that $q_i < q_{i+1}$, a straightforward application of Lemma \ref{lemma: inc-fns} gives
$$ \frac{kp_{i+1}}{2h} \leq (1 - q_i)^{k/h} - (1 - q_{i+1})^{k/h} \leq \frac{k}{h}p_{i+1}$$

Substituting this into the expression for $f_h(X)$ gives 
$$\frac{k\mathbb{E}(X)}{2h} = \sum_{i=1}^n \frac{kp_{i}x_i}{2h} \leq f_h(X) \leq \sum_{i=1}^k \frac{kp_ix_i}{h} = \frac{k\mathbb{E}(X)}{h} $$
\end{proof}

\begin{lemma} \label{lemma: top_vals_cond_ub}
For a random variable $X$, underlying sample space $[0,1]$, let $A$ be as in Definition \ref{def: top-values}. Then
$$\mathbb{E}(X|A) \leq 4f_h(X) $$
\end{lemma}

\begin{proof}
Splitting the the boundary value $x_t$ if necessary, assume $q_t= \frac{h}{2k}$. But then for $X'$ as in Lemma \ref{lemma: topvalbound}
$$ f_h(X') = \left(1 - (1 - q_1)^{k/h} \right)x_1 + ... + \left((1 - q_{t-1})^{k/h} - (1 - q_t)^{k/h} \right)x_t \leq f_h(X) $$
As $q_t = \frac{h}{2k}$, we can use Lemma \ref{lemma: inc-fns} (with $a = k/h$, $q_i \in [0, k/2h], i \leq t$) to get
$$ \frac{k}{2h}\mathbb{E}(X') \leq f_h(X') \leq f_h(X) $$ 
Also 
$$ \mathbb{E}(X|A) = \frac{1}{q_t} \left(\sum_{j=1}^t p_jx_j \right) = \frac{2k}{h}\left(\sum_{j=1}^t p_jx_j \right) = \frac{2k}{h}\mathbb{E}(X') $$
Therefore,
$$ \mathbb{E}(X|A) \leq 4f_h(X) $$
\end{proof}

In summary, we've seen that we can bound contributions of the top $h/2k$ quantile to $f_h$, and upper bound the contribution of the tail. We've also seen that we can upper and lower bound the expectation and the conditional expectation of $X$ using $f_h$.

\subsection{A Test with Constant Factor Approximation to Optimal}
Using the preliminary results we proved in the previous section, this section puts them together to give our main result:
\begin{theorem} \label{thm: main1}
If $X_1,...,X_k$ are the top scorers for the test function $f_h$, and $Y_1,...,Y_k$ is the true optimal team with respect to the team performance scoring function $g_h$, then for constant $\lambda$, ($\lambda < 30$),
$$g_h(Y_1,...,Y_k) \leq \lambda g_h(X_1,...,X_k) $$
\end{theorem}

The proof proceeds in two steps. First, we show an upper bound for $g_h$ in terms of $f_h$. In particular, if every member of the team $X_i$ has $f_h(X_i) \leq c$, we show that the team performance (according to $g_h$) is $\leq Ac$, where $A$ is a constant. After proving a similar lower bound, we can put the two together to get our desired constant factor approximation.  

\xhdr{The Upper Bound}

\begin{theorem} \label{thm: ub}
Let $X_1,...,X_k$ be random variables with $f_h(X_i) \leq c$. Then 
$$ g_h(X_1,...,X_k) \leq 2hc + \frac{hc}{\mre}$$
\end{theorem}

\begin{proof}
Assume the underlying sample space is $[0,1]^k$. Let $S \subset [k]$, and 
$$B_S = \{ \omega \in [0,1]^k : \omega_i > 1 - \frac{h}{2k} \iff i \in S \} $$
i.e. the event that $X_i$ takes values in its top $h/2k$ quantile iff $i \in S$. For a sample point $\omega \in B_S$, note that
$$(X_{X_1,...,X_k}^{(1)},...,X_{X_1,...,X_k}^{(h)})(\omega) \leq \sum_{i \in S} X_i(\omega) + \frac{hc}{\mre} $$
Indeed, if the top $h$ values are $X_{n_1},...,X_{n_h}$, with the first $m$, $n_1,...,n_m$ in $S$ then
$$ \sum_{ i = 1}^m X_{n_i}(\omega) \leq \sum_{i \in S} X_i(\omega)$$
The remaining random variables, $X_{n_{m+1}},...,X_{n_h}$ take tail values (as in Definition \ref{def: top-values}), so by Lemma \ref{lem: tail-bound},
$$ \sum_{i = m+1}^h X_{n_i}(\omega) < (h-m) \frac{c}{\mre} \leq \frac{hc}{\mre}$$
giving the inequality. Summing up over all $\omega \in B_S$, we get
$$g_h\left((X_1,...,X_k) 1_{B_S} \right) \leq \mathbb{E}\left( 1_{B_S} \sum_{i \in S} X_i \right) + \mathbb{P}(B_S) \frac{hc}{\mre}$$
But letting $A_i$ be the event that $\omega_i > 1 - \frac{h}{2k}$, and using independence of the $X_i$ and linearity of expectation
$$ \mathbb{E}\left( 1_{B_S} \sum_{i \in S} X_i \right) = \mathbb{P}(B_S)\sum_{i \in S}\mathbb{E}(X_i|A_i)$$
Using the bound in Lemma \ref{lemma: top_vals_cond_ub}, this becomes
$$ \mathbb{E}\left( 1_{B_S} \sum_{i \in S} X_i \right) \leq \mathbb{P}(B_S)|S|4c $$
Finally, as $\mathbb{P}(B_S) = \prod_{i \in S} \mathbb{P}(A_i) \prod_{i \notin S} (1 - \mathbb{P}(A_i))$, 
$$\mathbb{P}(B_S) = \left(\frac{h}{2k} \right)^{|S|} \left(1 - \frac{h}{2k} \right)^{k - |S|}$$
i.e. the number of $X_i$ taking their top values follows a Binomial distribution, parameters $(k, \frac{h}{2k})$. So, summing up over $B_S$ for all $S \subset [k]$, we get
$$ g_h(X_1,...,X_k) \leq \sum_{i=0}^k \binom{k}{i} \left(\frac{h}{2k}\right)^i \left(1 - \frac{h}{2k}\right)^{k-i} i\cdot4c + \frac{hc}{\mre}$$
Noting that the first term on the right hand side is just the mean ($h/2$) of the Binomial distribution scaled by $4c$ gives the result.
\end{proof}

\xhdr{The Lower Bound}
We now move on to a lower bound. We first give a lower bound for the case $h = 1$, when $g_h = \mathbb{E}(\max(\cdot))$, and show how to extend this for general $h$. To prove the $h= 1$ case, we will use our transformation in Lemma \ref{lemma: topvalbound} to zero all values lower than the top $1/2k$ quantile, and prove a lower bound on random variables with total positive probability mass $\leq 1/2k$. We thus first state and derive this.

\begin{lemma} \label{lem: rare_pot_lb1}
Let $X_1,...,X_k$ all have total positive probability mass $\leq \frac{1}{2k}$, with $f_1(X_i) \geq c$ for all $i$. Then
$$ \emax \geq 2c \left( \mre \right) $$
\end{lemma}

\begin{proof}
For any $X_i$, let $A_i$ be the event that $X_i$ is nonzero. We lower bound the expected maximum as follows: given $X_1,...,X_k$ in that order, we output the value of the first nonzero random variable we come across (starting from $X_1$ and finishing at $X_k$.)

This output value is pointwise less than or equal to the true maximum, so its expected value is a lower bound on the expected maximum. But its expected value is just
$$ \mathbb{P}(A_1) \mathbb{E}(X_1|A_1) + (1 - \mathbb{P}(A_1))\mathbb{P}(A_2) \mathbb{E}(X_2 |A_2) + ... + \left( \prod_{i=1}^{k-1} (1 - \mathbb{P}(X_i) \right) \mathbb{E}(X_k) $$
Noting that $\mathbb{P}(A_i) \mathbb{E}(X_i|A_i) = \mathbb{E}(X_i)$ and that $(1 - \mathbb{P}(A_i)) \geq (1 - \frac{1}{2k})$, we get
$$ \emax \geq \mathbb{E}(X_1) + \left(1 - \frac{1}{2k} \right) \mathbb{E}(X_2) + ... + \left(1 - \frac{1}{2k} \right)^{k-1} \mathbb{E}(X_k)$$ 
Using the lower bound of $\mathbb{E}(X_i) \geq \frac{f_1(X_i)}{k}$ from Lemma \ref{lem:rare_potbound}, summing up the geometric series, and noting $(1 - \frac{1}{2k})^k \geq ( \mre )$, we have
$$\emax \geq 2c \left( \mre \right)$$
as desired.
\end{proof}

We now prove our lower bound for $h = 1$. 

\begin{theorem} \label{thm: lb1}
Let $X_1,...,X_k$ be random variables with $f_1(X_i) \geq c$ for all $i$. Then
$$ \emax \geq 2c \left( \mre \right)^2$$
\end{theorem}

\begin{proof}
For any $X_i$ with total positive probability mass $ > \frac{1}{2k}$, we apply the transformation in Lemma \ref{lemma: topvalbound} to get $X'_i$, which is a lower bound on $X_i$. So certainly
$$\emax \geq \mathbb{E}(\max(X'_1,...,X'_k))$$
and by Lemma \ref{lemma: topvalbound}, 
$$f(X'_i) \geq c \left(\mre \right)$$ 
so using Lemma \ref{lem: rare_pot_lb1} , the statement of the theorem follows. 
\end{proof}

We now apply this to prove the main lower bound theorem

\begin{theorem} \label{thm : lb}
Let $X_1,...,X_k$ be random variables with $f_h(X_i) \geq c$ for all $i$. Then 
$$ g_h(X_1,...,X_k) \geq 2hc \left( \mre \right)^2 $$
\end{theorem}

\begin{proof}
Note that certainly
$$g_h(X_1,...,X_k) \geq \mathbb{E}(\max(X_1,...,X_{k/h})) + ... + \mathbb{E}(\max(X_{k - h + 1},...,X_k))$$
But each term on the right hand side is bounded below by $2c \left( \mre \right)^2$ by using Theorem \ref{thm: lb1}. So summing together, we have
$$g_h(X_1,...,X_k) \geq 2hc\left( \mre \right)^2$$ as desired.
\end{proof}

\xhdr{Finishing the proof}
With established lower and upper bounds, Theorem \ref{thm: main1} follows easily.

\begin{proof}
{\em (Theorem \ref{thm: main1})}
First note that if $l < h$, we can define $g_h(X_1,...,X_l)$ to be the sum of the expectations of all the $X_i$ as this is the same as adding $h - l$ random variables, each deterministically $0$.

Without loss of generality, 
let  $\{Y_1,...,Y_k \} = \{Y_1,...,Y_l, X_{l+1},...,X_k \}$ i.e. $X_{l+1},...,X_k$ is the intersection of the team formed of best test scorers and the optimal team. Now, if $c = \min_i f_h(X_i)$, then for $j \leq l$, as any $Y_j$ is not in the top $k$ scorers, $f_h(Y_j) \leq c$. 

Note that
$$ 2 g_h(X_1,...,X_k) \geq g_h(X_1,...,X_k) + g_h(X_{l+1},...,X_k) $$
Using the lower bound from Theorem \ref{thm : lb}, we get
$$ 2 g_h(X_1,...,X_k) \geq  2hc \left( \mre \right)^2 + g_h(X_{l+1},...,X_k)$$
On the other hand,
$$g_h(Y_1,..., X_{l+1},...,X_k)  \leq g_h (Y_1,...,Y_l) + g_h(X_{l+1},...,X_k)$$
Using the upper bound from Theorem \ref{thm: ub} then gives
$$g_h\left(Y_1,..., X_{l+1},...,X_k \right) \leq 2hc + \frac{hc}{\mre} + g_h (X_{l+1},...,X_k)$$
So we get that
$$g_h(Y_1,...,Y_k) \leq \lambda g_h(X_1,...,X_k)$$
where 
$$\lambda = \frac{ 2 \left(\mre \right) + 1}{ \left( \mre \right)^3} $$
\end{proof}

\subsection{A Different Test}
In the previous section we proved the main result of the paper, that there exists a test function, $f_h$, evaluating `potential', that can be used to select a team whose performance, according to a team performance function $g_h$, is only a constant factor from the optimal, independent of team size. 

A natural follow up question is whether $f_h$ is the only such test. From the proof, we can see that this is not the case.  If $E = \{ \omega : \omega > 1 - \frac{h}{k} \}$ for $\omega \in [0,1]$, the underlying sample space, then choosing $X$ according to the value of
$$ \mathbb{E}(X | E)$$
also provides a constant-factor approximation to the optimal set. 

\begin{theorem} \label{thm: other_test}
If $X_1,...,X_k$ are random variables with the $k$ highest values of $\mathbb{E}(X_i|E_i)$, where $E_i$ is the event that $X_i$ takes its top $h/k$ quantile of values, and $Y_1,...,Y_k$ is the optimal set size $k$, then for a constant $\mu$ independent of $k$,
$$ g_h(Y_1,...,Y_k) \leq \mu g_h(X_1,...,X_k)$$.
\end{theorem}

The two proofs are similar, which is expected, as the analysis of the function $f_h(\cdot)$ makes use of
quantities derived from $\mathbb{E}(X | E)$.
The function $f_h(\cdot)$ seems the more natural of the two,
however: it is arguably more direct to think about testing
an individual through repeated independent evaluations than to 
try quantifying what their top $h/k$ values are likely to be. The full proof is included in the Appendix.

\subsection{A Best Approximation?}
In this section we've seen that there exists a natural individual test, the potential test, that can get to within a constant factor ($\approx 30$) of optimal. We then outlined a different test (arguably slightly less natural to implement) which also gets to within a constant factor of the optimal ($\approx 16$). 

Seeing these constants, we might ask whether we can say something on whether there is some constant factor $C > 1$ which \textit{no} test can achieve. We prove that such a $C$ does indeed exist: 
\begin{theorem}
No test function $f$ can guarantee a constant factor approximation to the optimal closer than $9/8 = 1.125$ when evaluating team performance with the expected maximum.
\end{theorem}

\begin{proof}
Our proof is with a bad example. Assume we have three weighted Bernoulli random variables $X_1, X_2, X_3$ from which we wish to pick a team of size 2. A weighted Bernoulli random variable is one that takes exactly one nonzero value $v$ with some probability $p$, and can thus be characterized by the vector $(p, v)$. 

In that format, let our three Bernoulli random variables be $X_1 = (1/2, 2), X_2 = (1,1), X_3 = (1/2, 4/3)$. Note that $X_1$ is monotonically better than $X_3$, so any sensible test function $f$ should definitely pick $X_1$ and one of $X_2, X_3$. Indeed, if the team were to comprise of $(X_2, X_3)$, this would result in an expected maximum of $7/6$, a factor of $9/7$ away from the optimal team's expected maximum of $3/2$.

Breaking ties adversarially (as we can always perturb an example slightly in a tie), if $f(X_3) > f(X_2)$, then our team becomes $(X_1, X_3)$, but the expected maximum of this team is $4/3$, whereas the expected maximum of the team $(X_1, X_2)$ is $3/2$, and so $f$ is $9/8$ from optimal.

If on the other hand $f(X_2) > f(X_3)$, then consider a new triple of random variables $Y_1 = (1, 1), Y_2 = (1,1), Y_3 = (1/2, 4/3)$. As $Y_1, Y_2 = X_2$ and $Y_3 = X_3$, $f$ will pick the team $(Y_1, Y_2)$, which has an expected maximum of $1$ compared to picking a team of $(Y_1, Y_3)$ where the expected maximum is $7/6$, meaning $f$ is $7/6$ away from optimal. 

So the best \textit{any} test statistic can manage in this setting is a constant factor approximation of $9/8 = 1.125$. 
\end{proof}

\section{Submodularity and Negative Examples}
In this section, we recap properties of submodularity, prove the pointwise submodularity of $g_h$ and study the failure of the canonical test. We then more broadly look at submodular functions in general. We show that among submodular functions, the existence of an individual test function $f_h$ which can be used for a proof of constant factor optimality is an uncommon feature, relying on the unique properties of the expected maximum. 

\subsection{Submodularity, Pointwise Submodularity and the Canonical Test}

Earlier, we claimed that 
$\mathbb{E}(\max(\cdot))$ is {\em submodular}. 
In fact, a stronger statement is true. 
To state it, we recall our notation in which,
for a set $T$ of random variables, 
$X^{(j)}_T$ denotes the $j^{\rm th}$ largest in the set.

\begin{theorem} \label{thm: ptwise_submodular}
Let $\mathcal{U}$ be a large finite ground set of nonnegative random variables, with $\Omega$ being the underlying sample space. In a slight abuse of notation, for $\omega \in \Omega$, and $h \geq 0$, let
$$ \omega_h : \mathcal{P}(U) \rightarrow \mathbb{R}$$
be defined by
$$ \omega_h(T) = (X^{(1)}_T + ... + X^{(h)}_T)(\omega)$$
i.e. the sum of the top $h$ values of the random variables in $T$ evaluated at the sample point $\omega$. Then $\omega_h(\cdot)$ is submodular.
\end{theorem}

In summary, we prove that if $A = S \setminus \{y\}$, with $S \subset \mathcal{U}$, then for $x \notin S$, the submodular property
$$ \omega_h( S \cup \{x\}) - \omega_h(S) \leq \omega_h(A \cup \{x\}) - \omega_h(A)$$ 
holds.   We show this by fixing an order of elements in $S$ under $\omega$ and considering what each side of the inequality looks like. Chaining a set of inequalities of this form by removing one element each time gives the result for arbitrary subsets of $S$.

(Note that if $|A| < h$, only the first $|A|$ terms are possibly nonzero - we can increase $|A|$ by adding a number of deterministically zero random variables.)

\begin{proof}
{\em (Theorem \ref{thm: ptwise_submodular})}

Assume $S = \{X_1,...,X_n\}$, and  $A = \{X_1,...,X_{n-1}\}$. Rearranging, the submodularity inequality becomes
$$ \omega(A \cup \{X, X_n\}) + \omega(A) \leq \omega(A \cup \{X\}) + \omega(A \cup \{X_n\}) $$
First note that $X, X_n$ are interchangeable in the above inequality. We examine two cases.
\begin{itemize}
\item[(1)] At least one of $X, X_n$, wlog $X$ (by symmetry) is not in the top $h$ values in $\omega$. This has two easy subcases. If $|A| \geq h$, then
$$ \omega_h(A \cup \{X , X_n\}) + \omega_h(A) = \omega_h(A \cup \{X_n\}) + \omega_h(A)$$
and
$$ \omega_h(A \cup \{X\}) + \omega_h(A \cup \{X_n\}) = \omega_h(A) + \omega_h(A \cup \{X_n\})$$
so equality holds. In the other case, we have $|A| < h$, so we get
$$\omega_h(A \cup \{X\}) + \omega_h(A \cup \{X_n\}) = \omega_h(A) + X(\omega) + \omega_h(A) + X(\omega) $$ 
The left hand side of the target inequality becomes
$$\omega_h(A \cup \{X , X_n\}) + \omega_h(A) \leq (A + X + X_n)(\omega) + A(\omega) $$ 
with strict inequality if $|A| = h - 1$, as $X$ would be omitted in this case. So again, the desired inequality holds.

\item[(2)] Now, we may assume that $X_n, X$ are both in the top $h$. Assume $$X_n (\omega) = X^{(i)}_{A \cup \{X, X_n\}}$$ and 
$$X(\omega) = X^{(j)}_{A \cup \{ X , X_n \}}$$ 
and wlog $i > j$. In $A \cup \{ X , X_n\}$, let the top $h+2$ elements (with appropriately many zero elements) be ordered as below:
$$ X_{n_1}(\omega) \geq X_{n_2}(\omega) \geq ...X_{n_{i-1}}(\omega) \geq X_n(\omega) \geq X_{n_{i+1}}(\omega) \geq ... X_{n_{j-1}}(\omega) \geq X(\omega) \geq X_{n_{j+1}}(\omega) \geq ...  X_{n_{h+2}}(\omega)$$
Then we get 
$$ \omega_h(A \cup \{X , X_n\}) + \omega_h(A) = \left( 2 \left( \sum_{\substack{ l = 1 \\ l \neq i,j}}^{h-2} X_{n_l} \right) + X + X_n + X_{n_{h+1}} + X_{n_{h+2}} \right)(\omega) $$
and 
$$ \omega_h(A \cup \{X\}) + \omega_h(A \cup \{X_n\}) = \left( 2 \left( \sum_{\substack{ l = 1 \\ l \neq i,j}}^{h-2} X_{n_l} \right) + X + X_n + 2X_{n_{h+1}} \right)(\omega)$$
Noting that $X_{n_{h+1}} \geq X_{n_{h+2}}$ gives the result.
\end{itemize}
\end{proof}

A useful corollary is:

\begin{corollary}
For $h \geq 1$, $g_h( \cdot)$ is submodular. 
\end{corollary}
which follows from the theorem by taking expectations.

There are many results about the tractability (or approximate
tractability) of optimization problems associated with submodular functions.
For our purposes here, the most useful among these results is
the approximate maximization of arbitrary monotone submodular functions
over sets of size $k$.
This can be achieved by a simple greedy algorithm, which starts with the empty
set, and at each stage, iteratively adds the element providing the
greatest marginal gain; the result is a provable
$(1 - 1/e)$ approximation to the
true optimum \cite{nemhauser-hill-climbing}.  Note that this means we
can find a good approximation of the optimal set {\em even when the
random variables $X_i$ are dependent}. (See Section 4 for further
discussion of this.)

\xhdr{The Canonical Test}
In Section 2, our motivation for studying $f_h$, a measure of potential, was the failure of the \textit{canonical test}, selecting a team according to $\mathbb{E}(X)$. Here we use the property of submodular functions to prove the failure of this test.

\begin{observation}
If $f$ is a submodular function on $\mathcal{P}(U)$, then for every $S \subset U$
$$f(S) \leq \sum_{x \in S} f(\{x\})$$
\end{observation}

This naturally leads to:

\begin{proposition}
If $g_h(\cdot)$ is the team evaluation metric, with $Y_1,...,Y_k$ being the true optimal set, and $X_1,...,X_k$ the random variables with the $k$ highest expectations (with $\mathbb{E}(X_i) \geq \mathbb{E}(X_j)$ if $i \geq j$) then
$$ g_h(Y_1,....,Y_k) \leq \frac{k}{h} g_h(X_1,...,X_k)$$ 
and this bound is tight.
\end{proposition}

\begin{proof}
By the observation, we note that
$$g_h(Y_1,....,Y_k) \leq \sum_{i=1}^k g_h(Y_i) = \sum_{i=1}^k \mathbb{E}(Y_i)$$
But as $X_1,...,X_k$ are the elements with the $k$ highest expectations,
$$ \sum_{i=1}^k \mathbb{E}(Y_i) \leq \sum_{i=1}^k \mathbb{E}(X_i) \leq \frac{k}{h} \sum_{i=1}^h \mathbb{E}(X_i)$$
the last inequality following from the assumption on the ordering of the $X_i$. Finally, 
$$ g_h(X_1,...,X_k) \geq g_h(X_1,...,X_h) = \sum_{i=1}^h \mathbb{E}(X_i)$$
the last equality as there are only $h$ values. Putting it together, we have
$$g_h(Y_1,...,Y_k) \leq \frac{k}{h} g_h(X_1,...,X_h) \leq \frac{k}{h} g_h(X_1,...,X_k)$$
as desired. For tightness, let $X_i$ be deterministically $1 + \epsilon$ and $Y_i$ be $n$ with probability $1/n$ for large $n$. Then
$$ g_h(Y_1,...,Y_k) \geq \sum_{i=0}^h in \binom{k}{i} \left(\frac{1}{n} \right)^i \left(1 - \frac{1}{n} \right)^{k-i} \geq n \left(1 - \left(1 - \frac{1}{n}\right)^k \right) = k + O \left(\frac{1}{n} \right) $$
Also,
$$g_h(X_1,...,X_k) = h \left(1 + \epsilon \right)$$
So as $n \rightarrow \infty$ and $\epsilon \rightarrow 0$, we have
$$ g_h(Y_1,...,Y_k) \rightarrow \frac{k}{h} g_h(X_1,...,X_k)$$
\end{proof}

\subsection{Test Scores for Other Submodular Functions}
In the previous section, we saw that for $g = \mathbb{E}(\max(\cdot))$, a submodular function, we were able to define an individual test score with a constant factor approximation to the optimal. Furthermore we were able to define a \textit{family} of submodular functions $g_h$ interpolating between the expected maximum and a sum of expectations, which all had this property. It is therefore natural to wonder whether this is a property shared by many submodular functions. One way to formalize this question might be:

\begin{question}
Given a (potentially infinite) universe $U$, for which associated submodular functions $g$ does there exist a test score $f$
$$ f : U \rightarrow \mathbb{R}^+ $$
such that for any subset $S \subset U$, if $x_1,...,x_k \in S$ are the elements with the $k$ highest values of $f$, then $g(x_1,...,x_k)$ is always a constant-factor approximation to 
$$\max_{T \subset S, |T| = k} g(T)$$
% $g(T)$ with $T \subset S, |T| = k$?
\end{question}

Despite the positive result in Section 2, we find that many common submodular
functions depend too heavily on the interrelations between elements
for independent evaluations of elements to work well. We present two
such examples.

\xhdr{Cardinality Function}

One of the canonical examples of a submodular function is the set cardinality function. Let $U = \mathcal{P}(\mathbb{N})$. Then for $T = \{T_1,...,T_m \}$, with $T_i \in U$, 
$$g(T) = | \cup_{i=1}^m T_i  |$$
% $$g(T) = | \{ x : x \in T_i, 1 \leq i \leq m \} |$$
This function has a natural interpretation for team performance.
We can imagine each candidate as a set $T_i$, consisting of the
set of perspectives they bring to the task.
$g(T_1, T_2, \ldots, T_m)$ is then the total number of distinct
perspectives that the team members bring collectively;
this objective function is used in arguments that diverse
teams can be more effective
\cite{hong-page-diversity-pnas,marcolino-tambe-team-formation}.

We show a negative result for the use of test scores with this function.

\begin{theorem}
In the above setting, with universe $U$, and $g$ the set cardinality function, no such test score $f$ exists.
\end{theorem}

\begin{proof}
Suppose for contradiction such an $f$ did exist. Assume ties are broken in the worst way possible (no information is gained from a tie.) Let $U_1,U_2,...$ be disjoint intervals in  $\mathbb{N}$ with
$$ U_i = \{ (i-1)(k + 1) + 1,..., i(k + 1) \}$$
And let 
$$ V_i = \{ S \subset U_i : |S| = k \}$$
i.e. the set of all size $k$ subsets of $U_i$.
We will find it useful to label elements of $V_i$ based on their $f$ value, so let
$$ V_i = \{X_{i1},...,X_{ik+1} \} $$
with 
$$f(X_{i1}) \leq f(X_{i2}) ... \leq f(X_{ik+1})  $$
Call a set $V_j$, $j > k$ {\em bad} with respect to $V_1$ if 
$$f(X_{j1}) \leq f(X_{12})$$ 
and {\em good} otherwise.
Note that we cannot have more than $k$ $V_j$ bad with respect to $V_1$. Else, supposing $V_{n_1},...,V_{n_k}$ were all bad with respect to $V_1$, in the set 
$$ S = \{ X_{12},...,X_{1k+1}, X_{{n_1}1},...,X_{{n_k}1} \}$$
the $k$ set chosen by $f$ would be $X_{12},...,X_{1k+1}$, for a $g$ value of $k+1$, but the optimum is given by $X_{{n_1}1},...,X_{{n_k}1}$, for a $g$ value of $k^2$ - a factor of $ \approx k$ difference.

So there are at most $k$ bad sets with respect to $V_1$. But the same logic applies to $V_2,...,V_k$. So in $V_{k+1},...,V_{k^2 + k + 1}$ there is at least one set, say $V_{j}$, that is good with respect to $V_1,...,V_k$. But then in the set
$$S = \{ X_{11},...,X_{k1}, X_{j1},...,X_{jk} \}$$
the $k$ set chosen by $f$ would be $X_{j1},...,X_{jk}$, with a $g$ value of $k + 1$, but the optimum would be $X_{11},...,X_{k1}$ with a $g$ value of $k^2$. 
\end{proof}

\xhdr{Linear Matroid Rank Functions}
Another class of measures of team performance is given by
assigning each candidate a vector $v_i \in \mathbb{R}^m$,
and the performance of a team $v_1, v_2, \ldots, v_k$
is the rank of the span of the set of corresponding vectors.
Such a measure has a similar motivation to the previous
set cardinality example: if the team is trying to solve a
classification problem over a multi-dimensional feature space,
then $v_i$ may represent the weighted combination of features
that candidate $i$ brings to the problem,
and the span of $v_1, v_2, \ldots, v_k$ establishes the effective
number of distinct dimensions the team will be able to use.

More generally, the rank of the span of a set of vectors is 
a matroid rank function, and we can ask the question in that context.
Given a matroid $(V, \mathcal{I})$ and a set $S  \subset V$, the matroid rank function $g$ is
$$ g(S) = \max \{|T| : T \subset S, T \in \mathcal{I} \} $$
i.e. the maximal independent set contained in $S$. It is well known that matroid rank functions are submodular \cite{birkhoff-matroid-rank}.
To come back to our vector space example, 
we show that when our underlying set is $\mathbb{R}^m$, and $\mathcal{I}$ are subsets that are linearly independent, no single element test can capture the relation between vectors well.

\begin{theorem} \label{thm: matroid_rank_thm}
For $U, g$ as above, no test score with good approximation exists.
\end{theorem}

The proof of this theorem relies on the fundamental property of $\mathbb{R}$. We show that for any sequence along a specific direction, the $f$ values for this sequence must be {\em bounded}. By the defining property of $\mathbb{R}$, each sequence then has a convergent subsequence. Looking at these convergent subsequences along each of $k$ coordinate axes $e_1,...,e_k$, we can then pick our {\em bad set} fooling $f$ into choosing $O(k)$ points in the same direction. See the Appendix for a full proof.

\subsection{Result for a Supermodular Function}
The above two examples show bad cases for submodular functions. As is expected, {\em supermodular} functions also have a negative answer to Question 3.5. 

A classic example of a supermodular function is the {\em edge count} function. 
\begin{definition}
Given a graph $G = (V,E)$, and a set $S \subset V$, $g(S)$ is the number of edges in the induced subgraph with vertex set $S$. 
\end{definition}

It is easy to check 
that $g$ is supermodular. 
$g$ also forms our bad example for supermodular functions.

\begin{theorem} \label{thm: supermodular}
Let $U$ be a very large graph, containing at least $N$ disjoint complete graphs with $k+1$ vertices - i.e. $K_{k+1}$. Then there is no test score $f$ with a constant (independent of $k$) order approximation property to the optimal $k$ set with respect to $g$ 
\end{theorem}

The proof is very similar to the cardinality function case. In that, we wanted to avoid picking subsets of the same set; in this, we would like to pick as many vertices in a single clique as possible. We adjust the notion of {\em bad} accordingly to ensure this doesn't happen, and arrive at our desired contradiction identically to before.

A particularly interesting feature of this case, is that, without the canonical statistical test for submodular functions, we can have an arbitrarily bad approximation ratio - even if $f$ is defined to be constant on each vertex, the counterexample demonstrates that $f$ may pick a set with {\em no} induced edges.

\section{Hill Climbing and Optimality}

For most non-trivial submodular functions, finding the optimal
solution is computationally intractable. 
This is the case for the maximum of a set of random variables
that are not necessarily independent.
In particular, suppose that $S = \{X_1, X_2, \ldots, X_n\}$
is a set of dependent random variables.
For a set $T$ of them, we can define $g(T)$ to be the
expected maximum of the random variables in $T$.
We now argue that maximizing $g(T)$ is an NP-hard problem in general.
We will do this by reducing an instance of {\em Set Cover}
to the problem.

Recall that in set cover, we have a universe $U$, and a set $ T = \{S_1,..,S_n\}$ of subsets of $U$ i.e. $S_i \subset U$ for all $i$. We wish to know if there is a subset $T' \subset T$, with $|T'| \leq k$, such that 
$ \bigcup_{S_i \in T'} S_i = U.$
To model this with random variables, let the underlying sample space be $U$, and each $X_i = 1_{S_i}$ the indicator function for the set $S_i$. Then it is easy to see that there exists a team size $k$ with expected maximum $1$ if and only if there exists $T'$ as above, $|T'| \leq k$. So maximizing the expected maximum of a set size $k$ provides an answer to the NP complete decision problem.   

In terms of approximation, 
we can apply the general hill-climbing result
mentioned earlier \cite{nemhauser-hill-climbing} to 
provide a $(1 - 1/e)$ approximation for finding the 
set of $k$ dependent random variables with the largest expected maximum. 

A natural question is whether independence is a strong enough assumption to guarantee a better approximation ratio. Indeed, we may even be tempted to ask
\begin{question} \label{ques: hill-climb_max_general}
If $X_1,...,X_n$ are (discrete) independent random variables, does hill-climbing find the size $k$ set maximizing the expected maximum?
\end{question}

Unfortunately, this is false. For a simple counterexample, take $X$ taking positive values $(9/5, 6/5)$ with respective probability masses $(1/3,1/3)$, $Y$ deterministically $1 + \epsilon$ for $\epsilon$ very small, and $Z$ taking a positive value $3/2$ with probability $2/3$. Then 
$\mathbb{E}(Y) > \mathbb{E}(X), \mathbb{E}(Z)$
which means in the first step, hill-climbing would choose $Y$. But,
$$ \mathbb{E}(\max(X,Z)) > \mathbb{E}(\max(Y,Z)), \mathbb{E}(\max(X,Y))$$
so hill-climbing would not find the optimal solution. In this counterexample, $Y,Z$ are both examples of weighted Bernoulli random variables.

\begin{definition}
We say a random variable $X$ has the weighted Bernoulli distribution, if $X = x$ for some $x \geq 0$ with probability $p$, and $X = 0$ otherwise. 
\end{definition}

What is surprising is that when {\em all} our random variables are weighted Bernoulli, Question \ref{ques: hill-climb_max_general} has an affirmative answer.

\begin{theorem} \label{thm: opt_greedy}
Given a pool of random variables, each of weighted Bernoulli distribution, performing hill-climbing with respect to $\mathbb{E}(\max(\cdot))$ finds the size $k$ set maximizing the expected maximum. 
\end{theorem}

In the context of forming teams, we can think of candidates
with weighted Bernoulli distributions as having a sharply
``on-off'' success pattern --- they have a single way to succeed,
producing a given utility, and otherwise they provide zero utility.

For $X$ as above, we will find it convenient to denote $X$ as $(p,x)$. For two weighted Bernoulli random variables $X = (p,x)$ and $Y = (q,y)$, we use $X \geq Y$ to mean $x \geq y$.  For $X_i = (p_i, x_i)$, with $X_1 \geq .. \geq X_k$, the expected maximum has an especially clean form:
$$ \emax = p_1x_1 + (1 - p_1)p_2x_2 + ... \prod_{i=1}^{k-1}p_kx_k $$
Rewriting this slightly, it also has an intrinsically recursive structure
$$ \emax = p_1x_1 + (1 - p_1) \mathbb{E}(\max(X_2,...,X_k))$$

As a step towards proving Theorem \ref{thm: opt_greedy}, we need two useful lemmas on when random variables can be exchanged without negatively affecting the expected maximum. Assume from now on all random variables are weighted Bernoulli.

Our first lemma shows that if one random variable dominates another in both nonzero value and expectation, we may always substitute in the dominating variable. So given two random variables with the same expected value, we always prefer the 'riskier' random variable.

\begin{lemma} \label{lemma: large_val_better}
If $X \geq Y$, and $\mathbb{E}(X) \geq \mathbb{E}(Y)$, then for any $X_1,...,X_k$, 
$$ \mathbb{E}(\max(X,X_1,...,X_k)) \geq \mathbb{E}(\max(Y,X_1,...,X_k))$$
\end{lemma}

\begin{proof}
{ \em (Lemma \ref{lemma: large_val_better})}
Assume $X_i$ are in value order. Wlog assume $X \geq X_i$ for all $i$ (an almost identical proof works if that is not the case) and that $X_t \geq Y \geq X_{t+1}$. Letting $X = (p,x)$, $X_i = (p_i, x_i)$. Also, assume that $Y = (q, y)$. By the recursive structure of the expected maximum for weighted Bernoulli random variables,
$$\mathbb{E}(\max(X,X_1,...,X_k)) = px + (1 - p)b + (1 - p)sc$$
and that
$$ \mathbb{E}(\max(Y,X_1,...,X_k)) = b + sqy + s(1 - q)c $$
where
\begin{align*}
b &= \mathbb{E}(\max(X_1,...,X_t)) \\
s &= \mathbb{P}(X_1,...,X_t = 0) \\
c &= \mathbb{E}(\max(X_{t+1},...,X_k))\\
\end{align*}
Note $b + sc \leq x$ as $X \geq X_1,...,X_k$. So, if $p \geq q$, 
$$ px + (1 - p)(b + sc) \geq qx + (1 - q)(b + sc) $$
The left hand side of the above is just $\mathbb{E}(\max(X,X_1,...,X_k))$, so we can assume $p \leq q$ by decreasing $p$ to $q$ if necessary, and this will only decrease the value of $\mathbb{E}(\max(X,X_1,...,X_k))$. Now, note that 
$$ \mathbb{E}(\max(X,X_1,...,X_k)) \geq \mathbb{E}(\max(Y,X_1,...,X_k)) \iff px - pb - sqy + (q - p)sc \geq 0 $$
But $b/(1 -s)$ is a convex combination of $X_1,...,X_t$, so $b/(1 - s) \leq x$. So,
$$ px - pb - sqy + (q - p)sc \geq spx - sqy + (q - p)sc $$
Finally, by assumption, $\mathbb{E}(X) \geq \mathbb{E}(Y)$, and $p \leq q$, so the result holds.
\end{proof}

The next lemma describes a slightly technical variant of the above substitution rule: 

\begin{lemma} \label{lemma: wb_swap}
Let $X \geq Y$, and 
$\mathbb{E}(\max(X, X_1,...,X_k)) \geq \mathbb{E}(\max(Y,X_1,...,X_k)).$
Then if $Y_1,...,Y_m$ such that $Y \geq Y_i$ for all $i$, 
$$\mathbb{E}(\max(X, X_1,...,X_k,Y_1,...,Y_m)) \geq \mathbb{E}(\max(Y,X_1,...,X_k,Y_1,...,Y_m))$$
\end{lemma}

The proof of this lemma is similar to the first lemma and is in the Appendix.

We can now easily prove Theorem \ref{thm: opt_greedy}

\begin{proof}
{\em (Theorem \ref{thm: opt_greedy})}
We prove this inductively, showing that the element chosen by hill-climbing at time $i$ is part of the optimal set from then on.
Our base case is proving the first element chosen, $X = (x,p)$, which has greatest expectation, is always in the optimal set. Suppose the optimal set size $k$ is $\{ Y_1,...,Y_k\}$. Then if some $Y_i \leq X$, by Lemma \ref{lemma: large_val_better}, we could replace $Y_i$ by $X$. So $X \leq Y_k$. But as $Y_k$ only appears as $\mathbb{E}(Y_k)$ in $\mathbb{E}(\max(Y_1,...,Y_k)$, and $X$ has greatest expectation, we can replace $Y_k$ by $X$. 

Suppose we have chosen $t$ random variables, $X_1 \geq ... \geq X_t$, with the $t^{\mathrm{th}}$ random variable chosen being $X_i$. By the induction hypothesis, we know $X_j$ for $j \neq i$ are part of any $\geq t$ sized optimal set. For an optimal solution size $k$, let $Y_1 \geq ... \geq Y_m$ (where $m$ may equal $0$) be the random variables distinct from $X_i$, inbetween $X_{i-1}$ and $X_{i+1}$ value-wise. Similarly, let $Z_1\geq...\geq Z_h$ be the random variables inbetween $X_{i+1}$ and $X_k$. We have a few cases.

First note if $m > 0$, and $X_i \geq Y_j$ some $j$, then as $\mathbb{E}(\max(X_i,...,X_t)) \geq \mathbb{E}(\max(Y_j,X_{i+1},...,X_t))$, by applying Lemma \ref{lemma: wb_swap}, we can swap $Y_j$ with $X_i$. So $X_i \leq Y_j$ for all $j$, or $m = 0$. In either case, if $h > 0$, applying Lemma \ref{lemma: wb_swap} again, we may swap $X_i$ with $Z_1$. So $h = 0$, and so in order value, the final string of random variables in the optimal set is just $X_i, X_{i+1},...,X_k$. Note that if we take the smallest random variable distinct from the $X_l$ larger than $X_i$, say $Y$, $X_j \geq Y \geq X_{j+1}$, then as
$$\mathbb{E}(\max(X_1,...,X_t)) \geq \mathbb{E}(\max(Y,X_1,...,X_{i-1},X_{i+1},...X_t))$$
from the choice of elements by the hill-climbing algorithm, by the recursive structure of the expected maximum, we must have
$$ \mathbb{E}(\max(X_j,X_{j+1},...,X_i,...,X_t)) \geq \mathbb{E}(\max(X_j,Y,...,X_{i-1},X_{i+1},...,X_t))$$
so we can swap $Y$ with $X_i$. This completes the induction step, and the proof.
\end{proof}

This proof method gives us a simple condition which is sufficient (though slightly stronger than necessary) for when the hill climbing algorithm finds the optimal set:
\begin{condition} \label{condition:hill-climbing}
Let $f$ be a submodular function on a universe $U$. If $S_{t} = \{x_1,...,x_t \}$ is the set picked by hill climbing at time $t$, (with $S = \emptyset$) at $t = 0$, and $x_{t+1}$ is the next element chosen by hill climbing, then for any $ Z \subset U \setminus S_t$, must have
$$ \max_{z \in Z} f\left( S_t \cup \{x_{t+1} \} \cup Z \setminus \{z \}\right) \geq f\left(S_t \cup Z\right) $$ 
\end{condition}
For submodular functions satisfying Condition \ref{condition:hill-climbing},
it is possible to prove the optimality of hill-climbing as above.
Given that $S_t$ is part of the optimal set, we show that we can
always substitute in $x_{t+1}$ into the optimal solution and ensure
the value of $f$ doesn't decrease. Hence, $x_{t+1}$ must be part of
the optimal set.

\section{Test Scores for Competition}

Thus far we have considered a setting in which we want to assemble a
collaborative team, and we use test scores to identify team members.
But there are other natural contexts where we can ask about the power
of fixed ``scores'' to identify the quality of participants, and one of
these is a setting in which there is competition between individuals.

There is a large literature on the use of numerical scores to represent the
quality of participants in a competitive domain 
(e.g. \cite{elo-chess-ratings,herbrich-trueskill}). 
Our purpose in this short section is to describe a
basic result establishing a tight limit on the power of such scores in
an abstract setting.

We consider the following simple model of competition between pairs of
individuals.
Each possible competitor $i$ in our setting is represented by a random
variable $X_i$; we can think of $X_i$ as representing the distribution
of how well $i$ will perform in any given competition.
Thus, when competitors $i$ and $j$ are paired against each other,
each draws independently from their respective random variables 
$X_i$ and $X_j$;
these draws represent their performance in this instance of the 
$i$-$j$ competition. 
The competitor who draws the larger number is the winner.
(If they draw equal values, we declare them to have tied.)

Now, by analogy with previous sections --- but adapted here
to our competitive setting --- we would like to assign a numerical score 
to each competitor so that by comparing the scores of $i$ and $j$,
we can form an estimate of which is likely to win in a competition
between them.

A natural question is whether we can find a score for each competitor
so that the competitor with the higher score in a pairwise competition
is more likely to win.
Formulating this to allow for the possibility of ties as well, 
we'd like a function
$f$ that maps random variables to real numbers,
so that if $X_i$ and $X_j$ are random variables with 
$f(X_i) \geq f(X_j)$ then 
\[ \mathbb{P}(X_i \geq X_j) \geq \frac12. \]

It turns out that such a function does not exist.
To establish this fact, we use a counter-intuitive probabilistic
structure known as {\em non-transitive dice}.
A set of non-transitive dice is a collection of random variables
$X_1, \ldots, X_n$ for which $\mathbb{P}(X_i > X_{i+1}) > 1/2$
(with addition taken modulo $n$, so that $\mathbb{P}(X_n > X_1) > 1/2$
as well).

Here is a simple example, using six-sided dice $X, Y, Z$ with
non-standard sets of numbers written on their six faces.
Suppose
\begin{itemize}
\item $X$ has sides $2,2,4,4,9,9$;
\item $Y$ has sides $1,1,6,6,8,8$;
\item $Z$ has sides $3,3,5,5,7,7$.
\end{itemize} 
Then it is easy to compute that
\[ \mathbb{P}(X > Y) = \mathbb{P}(Y > Z) = \mathbb{P}(Z > X) = \frac{5}{9} \tag{*} \]

It is known that for all $\gamma < 3/4$, there exist sets of
non-transitive dice 
$X_1, \ldots, X_n$ for which $\mathbb{P}(X_i > X_{i+1}) > \gamma$
\cite{li-chien-nontransitive-dice,trybula-nontransitive-dice,usiskin-nontransitive-dice}.

Using non-transitive dice, one can directly put a limit on the
power of test scores for competition.

\begin{theorem}
Let $f$ be any function mapping random variables to real numbers,
and let $\beta > 1/4$.
Then there exist random variables $X$ and $Y$ for which
$f(X) \geq f(Y)$ but $\mathbb{P}(X \geq Y) < \beta$.
\label{thm:nontransitive}
\end{theorem}

\begin{proof}
Since $1 - \beta < 3/4$, we can find a set of 
non-transitive dice
$X_1, \ldots, X_n$ for which $\mathbb{P}(X_i > X_{i+1}) > 1 - \beta$.
For any function $f$ mapping random variables to real numbers,
let us apply $f$ to each of $X_1, \ldots, X_n$.
Let $f(X_i)$ be a maximum value among $f(X_1), \ldots, f(X_n)$.
Then we have $f(X_i) \geq f(X_{i-1})$ (since $f(X_i)$ is a maximum value),
but $\mathbb{P}(X_{i-1} > X_i) > 1 - \beta$ by the definition
of the sequence of non-transitive dice; and hence
$\mathbb{P}(X_i \geq X_{i-1}) < \beta$.
\end{proof}

Let us state this result in slightly different language.
A {\em test score} is any function 
$f$ mapping random variables to real numbers.
We say that $f$ has {\em resolution} $\alpha$ if 
for all random variables $X$ and $Y$ with $f(X) \geq f(Y)$,
we have $\mathbb{P}(X \geq Y) \geq \alpha$.
Then Theorem~\ref{thm:nontransitive} shows that there is no test score with
resolution $1/2$, and in fact no test score with resolution 
$\alpha$ for any $\alpha > 1/4$.

Suppose, then, that we were to weaken our goal and simply ask:
is there a test score with some positive resolution $\alpha > 0$?
We now show, via a simple construction, that this is the case:
in fact, there is a test score with resolution $1/4$, establishing
that the negative result of Theorem~\ref{thm:nontransitive} is tight.

\begin{theorem}
Let $f$ be a function that maps a random variable $X$ to a 
{\em median value} --- that is, a number $x$ such that
$\mathbb{P}(X \geq x) \geq 1/2$ and $\mathbb{P}(X \leq x) \geq 1/2$. 
(Note that such an $x$ need not be unique.)

Then if $X$ and $Y$ are random variables with $f(X) \geq f(Y)$, we have
$\mathbb{P}(X \geq Y) \geq 1/4$.
That is, $f$ is a test score with resolution $1/4$.
\end{theorem}

\begin{proof}
The proof follows directly from the definition of a median value.
Suppose $f(X) \geq f(Y)$. Then
\[ \mathbb{P}(X \geq Y) \geq \mathbb{P}(X \geq f(X)) \mathbb{P}(Y \leq f(Y) \geq \frac{1}{2} \cdot \frac{1}{2} = \frac{1}{4} \]
\end{proof}

\section{Conclusion and Open Problems}
In this paper, we have demonstrated that for a natural family of submodular
performance metrics, team selection can happen solely on an individual
basis, with minimal concession in team quality. However, this
selection criterion is more intricate than the canonical test
(singleton set value), the performance of which we also characterized.
Not all submodular functions are amenable to such an
approximation, and we exhibited examples where {\em no} function could
always guarantee a constant order bound. This leads to the natural
question of whether it is possible to characterize the {\em truly}
submodular functions (functions for which, like the expected maximum,
the canonical test performs poorly) which can approximated in such a fashion. 
There may be an opportunity to connect such questions to a distinct
literature on approximating a
submodular function with only a small number of values known
\cite{goemans-approx-submodular}, and approximation by juntas
\cite{submodular-juntas}. Another interesting direction is to relax the assumption of knowing the distribution of our random variables $X_i$. In many real life scenarios, we may not have a true skill distribution for candidates, but may instead have to rely on noisy samples. This problem may have links to work on robust estimation, \cite{huber-robust-estimation}.

Finally, we also explored the implications of independence of random
variables when using hill-climbing to approximate the size-$k$ set
maximizing the expected maximum. We established that for certain
random variables, we could find the true optimum this way. A natural
question is then, for what distributional assumptions can we guarantee
optimality, or a significantly better approximation ratio? Much work
has been done on structural properties of ensembles of random
variables with different distributions
\cite{k-modal}, \cite{poisson-binom},  and it is possible that such
techniques may be useful here.

\xhdr{Acknowledgments} 
This work was supported in part by a Simons Investigator Award, a Google Research Grant, a Facebook Faculty Research Grant, an ARO MURI grant, and NSF grant IIS-0910664.

% \input{jk-feb09b-bib}
% \bibliography{n,extra-refs}

\section{Appendix}
Here we provide a proof of \ref{thm: other_test}.

\begin{proof}
Note that if we find upper and lower bounds like Theorem \ref{thm: ub} and Theorem \ref{thm : lb}, then we can use the final part of the proof of Theorem \ref{thm: main1} unchanged to give our desired result. 

First, note that if $\mathbb{E}(X|E) \leq c$, then any value of $X$ not in its top $h/k$ quantile must be $\leq c$ (conditioning on $E$ ensures the the expectation of $X$ is a linear combination of the top values of $X$.) Now, if $X_1,...,X_k$ such that $\mathbb{E}(X_i|E_i) \leq c$ for all $i$, then letting $T \subset [k]$ and 
$$C_T = \{ \omega \in [0,1]^k : \omega_i > 1 - \frac{h}{k} \iff i \in T \}$$
be defined analogously to before, we get
$$g_h( (X_1,...,X_k )1_{C_T}) \leq \mathbb{P}(C_T)|T|c + \mathbb{P}(C_T)hc$$
as before. Summing up we note 
$$ \mathbb{P}(C_T) = \left( \frac{h}{k} \right)^{|T|} \left(1 - \frac{h}{k} \right)^{k - |T|}  $$
so we have a Binomial distribution parameters $(k, h/k)$, similar to before, so
$$g_h(X_1,...,X_k) \leq \sum_{i=0}^k ic \cdot \binom{k}{i}\left( \frac{h}{k} \right)^{i} \left(1 - \frac{h}{k} \right)^{k - i} + hc = 2hc$$
This gives us an upper bound. The lower bound is of a similar flavor to the upper bound. Suppose $X_1,...X_k$ such that $\mathbb{E}(X_i|E_i) \geq c$ for all $i$, and $T$ and $C_T$ are as above. Then note that
$$ g_h((X_1,...,X_k) 1_{C_T}) \geq \mathbb{P}(C_T) \cdot \min(|T|, h)c$$
i.e. for an event $\omega \in C_T$, $g_h(X_1,...,X_k)$ is greater than summing the minimum of $h$ and $|T|$ of the random variables that take values in their top $h/k$ quantile. Noting we have the same Binomial distribution as before
$$ g_h(X_1,...,X_k) \geq \sum_{i = h/2}^k \frac{hc}{2} \cdot \binom{k}{i}\left( \frac{h}{k} \right)^{i} \left(1 - \frac{h}{k} \right)^{k - i} \geq \frac{hc}{4} $$
where the last inequality follows by noting that as the mean of this distribution is $h$, the median certainly contained in the range $h/2 \leq i \leq k$. 

Note that to be entirely precise, we should replace $h/2$ with $\lfloor \frac{h}{2} \rfloor$. The $h = 1$ case then needs to be dealt with separately. For $h = 1$, note that the probability at least one of the $X_i$ takes a value in its top $h/k = 1/k$ quantile is
$$ 1 - \left(1 - \frac{1}{k} \right)^{k} \geq 1 - \frac{1}{e}$$
So for the $h = 1$ case we can bound below by
$$ \left(1 - \frac{1}{e} \right)c$$

We finish using the same proof as in Theorem \ref{thm: main1}, getting $\mu = 16$. 
\end{proof}

\yhdr{Submodularity and Negative Examples: Proofs}

We first give a proof of Theorem \ref{thm: ptwise_submodular}

Below is the full proof of Theorem \ref{thm: matroid_rank_thm}

\begin{proof}
{\em (Theorem \ref{thm: matroid_rank_thm})}
Like before, we assume for contradiction that such an $f$ does exist. We need a Lemma.
\begin{lemma} \label{lemma: bounded_along_line}
Let $x \in \mathbb{R}^m$. Then the set 
$$ \{ f(\lambda x) : \lambda \in \mathbb{R} \} $$
is bounded.
\end{lemma}
\begin{proof}
Suppose not, then there is a sequence $(\lambda_n)_{n \in \mathbb{N}}$ such that 
$$ f(\lambda_n x) \geq n$$
But letting $e_1, ..., e_k$ be the standard basis vectors, and $c = \max f(e_i)$, there are $\lambda_{n_1},...,\lambda_{n_k}$ with 
$$ f(\lambda_{n_i}x) > c$$
so in the set $ \{ e_1,...,e_k, \lambda_{n_1}x,...,\lambda_{n_k}x \}$, the optimal set has rank $k$ but the highest scoring $k$ set has rank $1$.
\end{proof}
The consequence (from the fundamental property of the real numbers) is that any sequence of vectors along a particular direction have a convergent subsequence. In particular, defining
$$ a_{in} = f \left(\frac{e_i}{n}\right) $$
we see that for each $i$, $(a_{in})$ has a convergent subsequence. Relabelling if necessary, let this convergent subsequence be $(a_{in})$, with
$$ a_{in} \rightarrow b_i$$
for each $i$. Wlog, we assume that $b_1 \geq b_2 ... \geq b_k$. We now complete the theorem by examining a few cases.
\begin{itemize}
\item[Case 1:] $b_1 > b_{k/2}$ In this case, we can take terms very close to $b_1$ and terms very close to $b_i$ for $i \geq k/2$ to ensure we pick all the $a_{1m}$ terms which only have rank 1. 

In more detail, let $\delta < b_1 - b_{k/2}$.  Then as we have a finite number of convergent sequences, $\exists N$ such that for all $m > N$, $|a_{im} - b_i| < \delta/3$ for all $i$. So for $l,m > N$, and for all $i \geq k/2$ we have
$$ a_{1m} > a_{il}$$
In particular, in the set
$$ \{ a_{1m},...,a_{1(m+k)}, a_{(k/2)l},...,a_{kl} \} $$
the $k$ set with the maximum $f$ values are the first $k$, for a rank of $1$, but the optimal set can achieve rank $k/2 + 1$ (taking say the last $k/2 + 1$ elements), providing the desired contradiction.

\item[Case 2] $b_1 = ... = b_{k/2} = b$ Here we derive a contradiction by looking more closely at what each sequence $a_{ij}$ for $i \leq k/2$ can do and deriving a contradiction. Assume from now on that $i \leq k/2$. 
\begin{itemize}
\item[(i)] If for some $i$, say $i = 1$, there was $n_1,...n_k$ and $\delta >0$ such that $a_{1{n_j}} > b + \delta$, then for $j \neq 1$, picking $a_{jl_j}$ within $\delta/2$ of $b$ would mean $\{ a_{1n_r} : r \leq k \} \cup \{ a_{jl_j} : j \leq k/2 \}$ would form a bad set for $f$, with a $2/k$ approximation ratio. 
\item[(ii)] So certainly only finitely many terms $ > b$ for any $i$. Discarding them, assume the sequences $a_{ij} \leq b$ for all $i,j$. If for some $i$, say $i = 1$, $k$ or more terms were equal to $b$, say $a_{1n_1},...,a_{1n_k}$ then for any $j$ (noting we break ties as in the worst case), $f$ performs poorly ($2/k$ approximation) on the set $\{a_{1n_1},...,a_{1n_k} \} \cup \{ a_{21},...,a_{(k/2)1} \}$.
\item[(iii)] So for each $i$, only finitely many terms $ = b$. Discarding those, assume all $a_{ij} < b$. Let $c = \min_i a_{i1}$. Then picking $n_1,...,n_k$ so $a_{1n_k} > c$,  $f$ has the same poor $2/k$ approximation on $\{ a_{11},...,a_{(k/2)1} , a_{1n_1},...,a_{1n_k} \}$.
\end{itemize}
\end{itemize}

This completes the proof of the Theorem.  $\square$
\end{proof}

We now give the full proof for the bad example for supermodular functions.

\begin{proof}
{ \em (Theorem \ref{thm: supermodular})}

Assume such an $f$ does exist. Let $K^1,...,K^N$ be the set of 
size-$(k+1)$ complete graphs. 
Let the vertices of $K^j$ be 
$\{v_{j1},...,v_{j(k+1)} \}$
in increasing
order of $f$-value, 
Consider $K^1$. For
$j > k$, say $K^j$ is {\em bad} with respect to $K^1$ if
$f(v_{j(k+1)}) \geq f(v_{1k})$. If $K^{n_1},...,K^{n_k}$ are all bad with
respect to $K^1$, then in the set
$\{v_{11},...,v_{1k},v_{n_1(k+1)},...,v_{n_k(k+1)} \}$, the set chosen
by the test score would be $v_{n_1(k+1)},...,v_{n_k(k+1)}$, for {\em no}
induced edges, while the optimal set is $v{11},...,v_{1k}$ with
$k(k-1)/2$ induced edges.

So there are less than $k$ graphs bad with respect to $K^1$. Similarly to before, applying the same argument to $K^2,...,K^k$,  we note that in $K^{k+1},...,K^{k^2 + k + 1}$, there is at least one graph that is not bad with respect to all of $K^1,...K^k$, say $K^m$. But then taking the set $\{v_{1(k+1)},...,v_{k(k+1)}, v_{m1},...,v_{mk} \}$, the test score pick $v_{1(k+1)},...,v_{k(k+1)}$ again with no induced edges, while the optimal set is $v_{m1},...,v_{mk}$ with $k(k-1)/2$ edges. 
\end{proof}

\yhdr{ Hill-Climbing and Optimality}

Below is the proof of the second lemma to show optimality in the weighted Bernoulli case.

\begin{proof}
{\em (Lemma \ref{lemma: wb_swap})}
We prove this by contradiction. Again, we may assume that $X \geq X_i$ for all $i$, $X_i$ are in value order, and $ X_t \geq Y \geq X_{t+1}$ as before. Using the notation of Lemma \ref{lemma: large_val_better} first note that $p \leq q$, as otherwise, $\mathbb{E}(X) \geq \mathbb{E}(Y)$, and we could directly apply Lemma \ref{lemma: large_val_better}. Our assumption gives the following inequality: 
$$ px + (1 - p)b + (1 - p) sc \geq b + sqy + (1 - q)sc $$

Suppose the Lemma is false. Then, we have
$$ px + (1 - p)b + (1 - p) sd < b + sqy + (1 - q)sd$$
where 
$$d = \mathbb{E}(\max(X_{t+1},...,X_k,Y_1,...,Y_m))$$
We show that both of these inequalities cannot hold simultaneously.

As $p \leq q$, we have that
$$ \mathbb{E}(\max(X,X_1,...,Y_m)) - \mathbb{E}(\max(X,X_1,...,X_k)) = (1 - p)s(d - c) \geq (1 - q)s(d - c)$$
 
But 
$$\mathbb{E}(\max(Y,X_1,...,Y_m)) - \mathbb{E}(\max(Y,X_1,...,X_k)) = (1 - q)s(d - c) $$
Writing
$$ \mathbb{E}(\max(X,X_1,...,Y_m)) = \left(\mathbb{E}(\max(X,X_1,...,Y_m)) - \mathbb{E}(\max(X,X_1,...,X_k)) \right) + \mathbb{E}(\max(X,X_1,...,X_k)) $$
and $\mathbb{E}(\max(Y,X_1,...,Y_m))$ analogously and comparing contradicts the falsity of the Lemma.
\end{proof}

\end{document}